\documentclass[11pt]{article}
\usepackage{booktabs} 
\usepackage[numbers]{natbib}
\usepackage{amsmath,amssymb,amsthm}
\usepackage[ruled]{algorithm2e} 
\usepackage{subcaption}
\usepackage{epigraph}
\usepackage{graphicx}
\usepackage{thmtools}
\usepackage{thm-restate}
\usepackage{hyperref}
\usepackage{fullpage}
\usepackage{cleveref}

\setcitestyle{acmnumeric}

\setcitestyle{acmnumeric}

\newcommand{\distD}{\mathcal{D}}
\newcommand{\rev}[1]{\mathrm{Rev}{(#1)}}
\newcommand{\detrev}[1]{\mathrm{DetRev}{(#1)}}
\newcommand{\srev}[1]{\mathrm{SRev}{(#1)}}
\newcommand{\brev}[1]{\mathrm{BRev}{(#1)}}
\newcommand{\bmrev}[1]{\mathrm{BuyManyRev}{(#1)}}
\newcommand{\bkrev}[2]{\mathrm{Buy}^{#2}\mathrm{Rev}{(#1)}}

\newcommand{\btworev}[1]{\mathrm{Buy2Rev}{(#1)}}
\newcommand{\lottery}{\vec{\mathrm{Lot}}}
\newcommand{\mech}{{\mathcal{M}}}

\newtheorem{theorem}{Theorem}
\newtheorem{observation}[theorem]{Observation}
\newtheorem{claim}[theorem]{Claim}
\newtheorem{lemma}[theorem]{Lemma}
\newtheorem{corollary}[theorem]{Corollary}

\newtheorem{obs}{Observation}
\newtheorem{definition}{Definition}
\newtheorem{example}{Example}
\newtheorem{question}{Question}
\newtheorem{conjecture}{Conjecture}

\newcommand{\gaptwo}[4]{\mathrm{gap}^{#1}_{#2} (#3, #4)}
\newcommand{\MenuGapTwo}[3]{\mathrm{MenuGap}^{#1} ( #2, #3 )}

\begin{document}

\title{Fine-Grained Buy-Many Mechanisms Are Not Much \\ Better Than Bundling}
\author{
  Sepehr Assadi\thanks{The authors would like to thank S. Matthew Weinberg and anonymous reviewers for providing feedback on the presentation of this paper.} \footnote{Department of Computer Science, Rutgers University.  Research supported in part by a NSF CAREER Grant CCF-2047061, and a gift from Google Research.}\\
  \texttt{sepehr.assadi@rutgers.edu}
  \and Vikram Kher\footnote{University of Southern California. This work was carried out while Vikram Kher was a participant in the 2022 DIMACS REU program at Rutgers University, supported by NSF grants CCF-1852215 and CNS-2150186.}\\
  \texttt{vkher@usc.edu}
  \and 
  George Li \footnote{University of Maryland. This work was carried out while George Li was a participant in the 2022 DIMACS REU program at Rutgers University, supported by NSF grants CCF-1852215 and CNS-2150186}\\ \texttt{gzli929@umd.edu} 
  \and
  Ariel Schvartzman  \footnotemark[1] \footnote{Google Research. Part of this research conducted while the author was at DIMACS, Rutgers University and was supported by the National Science Foundation under grant number CCF-1445755.} \\
  \texttt{aschvartzman@google.com}
}
\date{}

\maketitle

\begin{abstract}
    
    
    Multi-item revenue-optimal mechanisms are known to be extremely complex, often offering buyers randomized lotteries of goods. In the standard buy-one model, it is known that optimal mechanisms can yield revenue infinitely higher than that of any ``simple" mechanism---the ones with size polynomial in the number of items---even with just two items and a single buyer~\cite{BriestCKW15, HartN17}. 
        
    We introduce a new parameterized class of mechanisms, buy-$k$ mechanisms, which smoothly interpolate between the classical buy-one mechanisms and the recently studied buy-many mechanisms~\cite{CTT19, CTT20, CTT20b, RTT22}. Buy-$k$ mechanisms allow the buyer to buy up to $k$ many menu options. We show that restricting the seller to the class of buy-$n$ incentive-compatible mechanisms suffices to overcome the bizarre, infinite revenue properties of the buy-one model. Our main result is that the revenue gap with respect to bundling, an extremely simple mechanism, is bounded by $O(n^2)$ for any arbitrarily correlated distribution $\distD$ over $n$ items for the case of an additive buyer. Our techniques also allow us to prove similar upper bounds for arbitrary monotone valuations, albeit with an exponential factor in the approximation.
    
    On the negative side, we show that allowing the buyer to purchase a small number of menu options \emph{does not suffice} to guarantee sub-exponential approximations, even when we weaken the benchmark to the optimal buy-$k$ \emph{deterministic} mechanism. If an additive buyer is only allowed to buy $k = \Theta(n^{1/2-\varepsilon})$ many menu options, the gap between the revenue-optimal deterministic buy-$k$ mechanism and bundling may be exponential in $n$. In particular, this implies that no ``simple" mechanism can obtain a sub-exponential approximation in this regime. 

\end{abstract}

\section{Introduction}

How should a revenue-maximizing seller price an item for sale when facing a buyer with a private value for the item? If the seller knows the distribution of values, seminal work of Myerson~\cite{Myerson81} showed that it is optimal for the seller to offer the item at a take-it-or-leave-it price. The answer to this question becomes unclear for the case of multiple, even just two, items.

Optimal multi-item  auctions are known to be complex objects, offering no discernible mathematical structure and often exhibiting ``intuition-defying" properties~\cite{Daskalakis15}. A particularly egregious one is that there exist correlated distributions over just two items such that the revenue-optimal mechanism is \emph{infinitely} better than any ``simple" mechanism, ruling out the possibility of good worst-case approximations~\cite{HartN13, HartN17}.\footnote{By ``simple" mechanisms, we mean mechanisms of size polynomial in the number of items.} The bizarre aspect of these pathological distributions is that the optimal revenue is \emph{unbounded}, but any finite-sized mechanism can get at most finite revenue. One possible explanation for this bizarre phenomenon is that the seller is unrestricted  in their choice of mechanism: they only need to guard against the buyer's deviations towards any \emph{single} other allocation. This allows the seller to utilize mechanisms where the buyer can only purchase a single mechanism entry. These \emph{buy-one} mechanisms can be heavily tailored to the buyer's distribution, often offering comparable allocations for widely different prices. Consider the following  example.  


\begin{example}
\label{ex:1} A buyer walks into a coffee shop. They are equally likely to have one of three valuations over a cup of coffee and a bagel: either the buyer has value \$2 for the cup of coffee and \$0 for the bagel, \$0 for the cup of coffee and \$4 for the bagel, or \$4 for the cup of coffee and \$6 for the bagel (and \$10 for the combination of a cup of coffee and a bagel). The optimal mechanism in this example is as follows: the seller will offer the cup of coffee at \$2, the bagel at \$4 and the combination of a cup of coffee and a bagel at \$8. In this example, the optimal mechanism is buy-one incentive-compatible. The buyer with non-zero valuations for both items (weakly) prefers buying the combination at \$8 to buying exactly one of the items separately. The mechanism, however, is not buy-many incentive-compatible: when the buyer has non-zero value for both items, they would prefer to visit the coffee shop twice and buy the items separately for a combined price of \$6. This achieves the same allocation at a cheaper price. 
\end{example}

The example above highlights two problems with the classical buy-one model. The first is that no high-valued customer would pay \$8 for the combination of coffee and a bagel. They would buy one item, queue in line again, and buy the other. This in turn creates the second problem: the revenue of the optimal buy-one mechanism overshoots the real-world revenue the seller would experience. The buy-one mechanism would net the seller an expected revenue of $\$(2+4+8)/3 = \$4\frac{2}{3}$. In reality, because no buyer would pay $\$8$ for the combination of items, the seller would experience expected revenue $\$(2+4+6)/3 = \$4$. 

While this example is simple, the pathological constructions of~\cite{HartN17, PSW22} do significantly wilder things such as offering similar randomized allocations for astronomically different prices. For instance, a buy-one mechanism may offer randomized allocation $\vec{q}$ for price $p$, and randomized allocation $\vec{q}+\vec{\varepsilon}$ for price $4p$. Just like in Example~\ref{ex:1}, a buyer would prefer to buy the cheaper option two, three or even four times rather than the carefully tailored, more expensive option. However, if they are forced to buy exactly one option, these pathological constructions ensure high-valued buyers will marginally prefer buying one copy of the expensive item to buying one copy of the cheaper one. Thus part of the reason why positive results in multi-item auctions (especially for correlated items) are rare is because the ``optimal" buy-one mechanism is unrealistic and not implementable in a world were buyers may, reasonably, interact with the seller multiple times. It is this lack of consideration on the seller's choice of mechanism that allows for ``infinite" revenue auctions. 


One natural way to overcome this problem is to allow the buyer to purchase multiple menu entries. \emph{Buy-many} mechanisms, introduced more than ten years ago in~\cite{BriestCKW10, BriestCKW15}, are mechanisms where the buyer may purchase \emph{any} multi-set of menu entries, including sets of unbounded size. This significantly restricts the seller's choice of mechanism: buy-many mechanisms are always buy-one incentive compatible but the converse is not true. A simple way to see this is that the prices in (deterministic) buy-many mechanisms are always sub-additive, meaning that for any two disjoint sets of items $S, T$, $p(S)+p(T) \geq p(S\cup T)$. Buy-one mechanisms, like that of Example~\ref{ex:1}, need not satisfy this property, making them less appealing for real-world applications. 

The work of~\cite{BriestCKW10, BriestCKW15} already exhibits how buy-many mechanisms overcome the revenue gap problem: they showed that a popular benchmark, known as item-pricing, could recover an $O(\log n)$ factor of the revenue attained by the optimal buy-many mechanism for the case of a single, unit-demand buyer. This was later extended to arbitrary valuations by~\cite{CTT19}, while preserving the approximation factor. Key to these results is that by sufficiently restricting the seller's choice of mechanisms, the optimal revenue drops from \emph{unbounded} in the buy-one case to \emph{finite} in the buy-many case, allowing for simple mechanisms like item-pricing to approximate the optimal buy-many revenue. 

One question left unaddressed by these works is \emph{how much} we need to restrict the seller's choice of mechanisms so that the optimal revenue is finite. While in Example~\ref{ex:1} it was reasonable to assume the buyer would purchase a single item, re-queue and purchase the other item, it would not be reasonable to assume the buyer would be willing to re-queue \emph{any} number of times. At some point, the buyer will get tired. This means the buyer's threat of interacting with the mechanism any number of times is not a credible threat to the seller. Alternatively, we can think of buy-many mechanisms as giving \emph{too} much power to the buyer just like buy-one mechanisms give \emph{too} much power to the seller. 

In order to answer the question outlined we need a more fine-grained family of mechanisms that smoothly interpolates between buy-one and buy-many mechanisms. For this purpose we introduce buy-$k$ mechanisms, a parametric family of mechanisms where the buyer is allowed to purchase any multi-set of at most $k$ menu entries non-adaptively.\footnote{In other words, the buyer first chooses any multi-set of up to $k$ menu options and only after they commit any randomized allocations are decided. Our results will hold even if the buyer is allowed to adaptively choose the menu entries. See Appendix~\ref{sec:adapt} for a more detailed discussion.} We say a mechanism is buy-$k$ incentive-compatible if the buyer always prefers to buy a single menu entry rather than any multi-set of up to $k$ menu entries. Let $\mathcal{B}_{k}(\distD)$ be the set of buy-$k$ incentive-compatible mechanisms for a distribution $\distD$ over $n$ items, and let $\bkrev{\distD}{k} = \max_{\mech \in \mathcal{B}_k(\distD)} \rev{\distD, \mech}$ be the optimal revenue attainable by a buy-$k$ incentive-compatible mechanism. A simple observation, stated below and whose proof we defer to later in the paper, shows that as $k$ increases, the revenue of the seller weakly decreases. 
 
\begin{obs}
\label{cl:bmrev}
\begin{equation*}
\bkrev{\distD}{1} \geq \bkrev{\distD}{2} \geq \dots \geq \bmrev{\distD} \geq \brev{\distD}.
\end{equation*}
\end{obs}

We first show that $k$, the number of times the buyer might interact with the mechanism, can play a role in the seller's optimal revenue. We prove there exists a support-size $3$, correlated distribution $\distD$ over two items for which $\bkrev{\distD}{1} > \bkrev{\distD}{2} > \bkrev{\distD}{3} > \bkrev{\distD}{4}$. This example proves a strict separation between the class of buy-one and buy-$2$ mechanisms, and by Observation~\ref{cl:bmrev}, between the class of buy-$2$ mechanisms and buy-many mechanisms. This reinforces the idea that buy-$k$ mechanisms interpolate between buy-one and buy-many mechanisms, and thus merit a study of their own. 

\begin{restatable}{proposition}{propSep}
\label{prop:separation}

There exists a distribution $\distD$ over two items such that 

$$\bkrev{\distD}{1} > \bkrev{\distD}{2} > \bkrev{\distD}{3} > \bkrev{\distD}{4}. $$
\end{restatable}

We conjecture in fact that for the distribution from Proposition~\ref{prop:separation}, the seller's revenue strictly decreases as $k$ increases. In other words, we conjecture there exists a simple distribution that can witness separations between the classes of buy-$k$ and buy-$(k+1)$ mechanisms for all $k \geq 1$. We discuss this conjecture in Appendix~\ref{app:conj}. 

\begin{conjecture}
\label{conj}
There exists a distribution $\distD$ over two items such that for all $k \geq 1$, 
\[
\bkrev{\distD}{k} > \bkrev{\distD}{k+1}.
\]
\end{conjecture}

After proving that the class of buy-$k$ mechanisms is distinct from the previously studied classes of buy-one and buy-many mechanisms, the next natural question is to understand their approximation guarantees with respect to simple mechanisms. Our measure of simplicity for a mechanism $\mech$ will be its \emph{menu complexity} or the number of menu entries $|\mech|$ the mechanism offers. Under this lens, broadly speaking, we think of ``simple" mechanisms as those that have polynomial menu complexity and ``complex" mechanisms as those that have super-polynomial menu complexity. For example, any mechanism which only offers the grand bundle of all items for a fixed price has menu complexity $1$. This family of mechanisms is so important that the revenue of the optimal grand bundling mechanism, $\brev{\cdot}$ (henceforth bundling), is often a benchmark of interest.\footnote{Bundling is arguably one of the simplest mechanisms.} Therefore, our main question of interest is the following. 

\begin{question}
\label{oq:1}
Given integers $n$, $k$, when does $f(n, k) \cdot \brev{\distD} \geq  \bkrev{\distD}{k}$ hold for all distributions $\distD$ over $n$ items, for some function $f(n,k)$? 
\end{question}


\subsection{Our Contributions}

Our main result shows that restricting the seller to the class of buy-$n$ incentive-compatible mechanisms suffices to get around pathological constructions for two or more items (like e.g.,~\cite{HartN17,PSW22}). These works show that there are distributions over just two items for which no ``simple" mechanism could approximate the revenue of the optimal buy-one mechanism, or in the language of Open Question~\ref{oq:1}, that no such function $f(n,1)$ exists for $n \geq 2$. We show that when facing a single, additive buyer, the revenue from optimally pricing the bundle of items, $\brev{\distD}$, recovers a polynomial fraction of the optimal buy-$n$ revenue.




\begin{restatable}{theorem}{thmMainN}
\label{thm:MainN}

For any distribution $\distD$ over $n$ items for a single, additive buyer, it holds that 

$$O(n^2) \cdot \brev{\distD} \geq \bkrev{\distD}{n}.$$
\end{restatable}


The proof of~\ref{thm:MainN} relies on the identification of a measure, $\MenuGapTwo{k}{\cdot}{\cdot}$, whose formal definition we defer to Section~\ref{sec:Notation}. This quantity is the generalization to buy-$k$ mechanisms of $\MenuGapTwo{}{\cdot}{\cdot}$ introduced by previous work for buy-one mechanisms (see~\cite{HartN17, PSW22}). In those works, $\MenuGapTwo{}{\cdot}{\cdot}$ was used to construct distributions whose optimal revenue was hard to approximate. In contrast, our work is the first to show that this framework can be used to prove approximation guarantees instead. In fact, our techniques can be extended to also show similar results for the case of arbitrary monotone valuation functions,~\footnote{A valuation $x$ function is monotone if $x(S) \geq x(T)$ whenever $T \subseteq S$. In other words, getting more items never decreases the buyer's value.} (albeit with a significant loss in the approximation factor). 

\begin{restatable}{theorem}{thmGeneralVal}
\label{thm:GeneralVal}

For any distribution $\distD$ over $n$ items, for a single buyer with a monotone valuation, it holds that 

$$O(n^2 \cdot 2^n) \cdot \brev{\distD} \geq \bkrev{\distD}{n}.$$
\end{restatable}

We consider Theorem~\ref{thm:GeneralVal} as a result validating the robustness of the framework we introduce as a proof technique for approximation algorithms for revenue maximization. The experienced reader will recall that, historically, approximation algorithms for different valuations classes used tools specific to the valuations themselves (e.g., ~\cite{ChawlaHK07, ChawlaHMS10} for unit demand, \cite{Yao15, HartN12, BabaioffILW14} for additive valuations and so on). It wasn't until the work of~\cite{CaiDW16} that a unifying framework was developed to reprove (or even improve) such results. Thus, we interpret Theorem~\ref{thm:GeneralVal} additionally as proof that the framework we develop is robust enough to handle general valuation classes and are optimistic that results similar to Theorem~\ref{thm:MainN} can be proved via our framework for other valuation classes. 

The first piece of the proofs for Theorems~\ref{thm:MainN},~\ref{thm:GeneralVal} is identical. We show that there exists an appropriate choice of inputs $(X,Q)$ such that $\MenuGapTwo{k}{X}{Q}$ upper bounds the ratio between the optimal buy-$k$ revenue and the revenue achieved by bundling, up to some $O(k)$ factor. We again defer a technical definition of $\MenuGapTwo{k}{\cdot}{\cdot}$ until later. For the purposes of this exposition, it suffices to say that given two sequences of vectors $(X, Q)$, $\MenuGapTwo{k}{X}{Q}$ captures some geometric property of the input pairs of sequences. Thus, more precisely, in the first step of the proof we show that for any distribution $\mathcal{D}$ and any buy-$k$ incentive-compatible mechanism $\mech$ for $\distD$, there exists a cleverly chosen set of valuations $X$ in the support of the distribution $\distD$ together with their corresponding allocations $Q$ under $\mech$ whose ``geometric property" witnesses an upper bound to the  revenue revenue that $\mech$ achieves on $\distD$, up to a factor of $O(k)$. In particular, this is also true about the revenue-optimal buy-$k$ incentive-compatible mechanism. 

The second step of the proof upper bounds $\MenuGapTwo{n}{X}{Q}$ itself by $n$ for the case of an additive buyer (resp. by $n \cdot 2^n$ for the case of a monotone buyer) for \emph{all} input pairs $(X,Q)$. It is worth noting that for the case of an additive buyer, this second step is tight. This implies that our approach of bounding revenue gaps via $\MenuGapTwo{k}{X}{Q}$ cannot give a sublinear approximation. However, this is not a fault of our techniques. In Appendix~\ref{app:polyn} we provide a simple proof that there exist distributions $\distD$ for which $\brev{\distD} \leq \bkrev{\distD}{k}/O(n)$ for any $k$. In other words, $\brev{\distD}$ can not give a sublinear (in the number of items $n$) approximation to $\bkrev{\distD}{k}$, for any $k$. 


There are two subtle implications of these results. The first is that for all $n$-dimensional distributions $\distD$, $\bkrev{\distD}{n}$ is \emph{finite} whenever $\brev{\distD}$ is finite. This stands in contrast to the buy-one case where even for just $n=2$ items, there exist $\distD$ such that $\bkrev{\distD}{1} > \infty$ but $\brev{\distD} = O(1)$. The second is that  since $\bkrev{\distD}{k} \geq \bkrev{\distD}{k'}$ whenever $k < k'$ (due to Observation~\ref{cl:bmrev}), then Theorems ~\ref{thm:MainN},~\ref{thm:GeneralVal} in fact also answer Question~\ref{oq:1} for the case when $n \leq k$. 

The next goal is to answer Question~\ref{oq:1} for the case when $1 < k < n$. We make progress by proving a strong lower bound for the case $k \leq n^{1/2-\varepsilon}$. We show that there exist distributions for which there is an exponential revenue gap when $k \leq n^{1/2-\varepsilon}$. This is captured by Theorem~\ref{thm:mainLB}.


\begin{restatable}{theorem}{thmMainLB}
\label{thm:mainLB}
If $k \leq n^{1/2-\varepsilon}$ for some $\varepsilon > 0$, then there exists an additive valuation function $\distD$ over $n$ items such that for a single buyer

$$ \frac{\bkrev{\distD}{k}}{\brev{\distD}} \geq \frac{ \exp \left(\Omega{(n^{\varepsilon})}\right)}{2n^2}. $$ 
\end{restatable}

The experienced reader will observe that Theorem~\ref{thm:mainLB} says something strong about the revenue guarantees that simple mechanisms can obtain. Due to a Corollary from~\cite{HartN17} which says that bundling always recovers a $1/|\mech|$ fraction of the revenue of any mechanism $\mech$, the revenue of any mechanism $\mech$ of size $\textsc{poly}(n)$ cannot exceed $\textsc{poly}(n) \cdot \brev{\distD}$. Thus, Theorem~\ref{thm:mainLB} implies that for the instance that witnesses its proof, no mechanism of size polynomial in the number of items can obtain a sub-exponential approximation. The proof of Theorem~\ref{thm:mainLB} will, unsurprisingly, borrow ideas from~\cite{HartN17, BriestCKW15}. Interestingly the buy-$k$ mechanism used in the lower bound instance will be \emph{deterministic}, in part because the construction of the instance itself makes use of discrete combinatorial objects known as cover-free sets. This implies that the lower bounds hold even for the optimal deterministic buy-$k$ mechanism, a much weaker benchmark for comparison. To put this in context,~\cite{HartN13} showed that there exist distributions $\distD$ for which $\brev{\distD} \leq \frac{2^n}{n} \detrev{\distD}$, where $\detrev{\distD}$ is the revenue of the optimal buy-one deterministic mechanism. Thus, Theorem~\ref{thm:mainLB} can be seen as an extension of this result, proving that even weakening the seller's benchmark to the optimal buy-$k$ deterministic mechanism (for $k \leq n^{1/2-\varepsilon}$) does not significantly improve the worst-case approximation with respect to bundling. 

While our model and results are written for the case of a non-adaptive buyer, a simple argument will allow us to translate both our upper bounds (Theorems~\ref{thm:MainN},~\ref{thm:GeneralVal}) and our lower bounds (Theorem~\ref{thm:mainLB}) to the case of adaptive buyers. We defer this discussion to Appendix~\ref{sec:adapt}. 

\textbf{Our Techniques.} The main conceptual contribution of this work is the introduction of a new class of mechanisms, buy-$k$ mechanisms, which interpolate between buy-one and buy-many mechanisms. The main technical contribution of our work is a novel framework for proving approximation results for multi-item mechanism design under arbitrary distributions. We generalize measures meant for the buy-one setting from ~\cite{HartN17, PSW22} to the buy-$k$ setting. Similar to~\cite{PSW22}, we prove that this measure upper bounds the revenue gap between the revenue-optimal mechanism (in some class of mechanisms) and bundling. Unlike~\cite{PSW22}, we are able to show a \emph{finite} upper bound for this measure in the case of buy-$n$ mechanisms, yielding a finite approximation result.

\subsection{Related Work}

Buy-many mechanisms have been proposed more than ten years ago, as early as~\cite{BriestCKW10, BriestCKW15}. Results from a recent line of work~\cite{CTT19, CTT20, CTT20b, RTT22} make the case to further the study of buy-many mechanisms. For instance,~\cite{CTT20} show that buy-many mechanisms satisfy some form of revenue monotonicity, an intuitive property that does not hold in the case of buy-one mechanisms~\cite{HartR15, PSW19}. In addition, as mentioned earlier in the introduction,~\cite{CTT19} show that item-pricing recovers a $O(\log n)$ factor of the optimal buy-many revenue. Combining a Corollary from~\cite{HartN17} with the main result of~\cite{CTT19} shows that bundling recovers at least a $O(n \log n)$ fraction of the optimal buy-many revenue. Our results have a worse approximation factor because the benchmark is stronger (and this is proved formally in Observation~\ref{cl:bmrev}). Thus our work deepens the study of buy-many mechanisms by introducing more fine-grained classes of mechanisms. We believe our results strengthen the case for the study of not only buy-many mechanisms, but also fine-grained buy-many mechanisms. 

The work of~\cite{CTT19} also gave strong lower bounds for the \emph{description complexity}, a measure that lower bounds the menu complexity of a mechanism. In particular, they showed that no mechanism with sub-exponential description complexity could get an $o(\log n)$ approximation to the optimal buy-many revenue, even for additive buyers. In follow up work,~\cite{CTT20} extended the lower bound to the larger class of fractionally sub-additive (or XOS) valuations.

A prolific line of work assumes that the underlying distribution of values $\distD$ is a product distribution. Under this assumption, it is known that mechanisms with low menu complexity can achieve constant-factor approximations to the optimal revenue for sub-additive valuations (see e.g.,~\cite{LiY13, HartN12, Yao15,ChawlaHK07, ChawlaMS15, BabaioffILW14, CaiZ17, RubinsteinW15, ChawlaHMS10, ChawlaM16, BabaioffGN17}, among others), effectively circumventing the pathological constructions of~\cite{HartN17}. Some recent results even show strong positive results for arbitrary approximation schemes. For instance,~\cite{BabaioffGN17} show that for any product distribution $\mathcal{D}$, there exists a mechanism with finite menu complexity that recovers a $(1-\varepsilon)$ approximation to the optimal revenue when selling to an additive buyer. More recently,~\cite{KothariMSSW19} give a quasi-polynomial approximation scheme for revenue maximization for a single, unit-demand buyer interested in $n$ independent items. Notwithstanding the significant contributions of these works, the question of revenue-maximization under arbitrary distributions remained unaddressed. 

Finally, work of~\cite{PSW19} provides yet another way to circumvent the pathological constructions of~\cite{HartN17}. In their work, the authors borrow ideas from the celebrated smoothed-analysis framework and initiate the study of beyond worst-case revenue maximization. Their results show that, under some smoothing models, simple mechanisms can approximate optimal ones. 

\textbf{Organization.} In Section~\ref{sec:Notation}, we present formal definitions for the objects of our interest as well as for the relevant benchmarks we use. Section~\ref{sec:main} contains the proofs of our upper bounds in Theorems~\ref{thm:MainN},~\ref{thm:GeneralVal}. Section~\ref{sec:few} contains the proof of our lower bound in Theorem~\ref{thm:mainLB}. We conclude in Section~\ref{sec:conclusion} and outline questions for future work. Appendix~\ref{sec:adapt} contains the discussion for the case of adaptive buyers. Appendix~\ref{app:polyn} proves that Theorem~\ref{thm:MainN} can not be improved to provide a sub-linear approximation. Appendix~\ref{app:conj} presents the proof of Proposition~\ref{prop:separation} and presents a simple candidate distribution to prove Conjecture~\ref{conj}.

\section{Preliminaries and Notation}
\label{sec:Notation}

We consider the case of a single buyer interested in $n$ items from a single revenue-maximizing seller. We assume the buyer is utility-maximizing and risk-neutral. The buyer draws their valuation $x$ from a known, possibly correlated $n$-dimensional distribution $\distD$, with support set $\mathcal{X} = \textsc{supp}({\distD})$ and probability density function $f_{\distD}(\cdot)$ (and the index may be ommitted when clear from context). For all our results, we will assume that the buyer's valuation over the items is monotone, i.e., whenever $S\subseteq T$, ${x}(S) \leq {x}(T)$. For our main result, Theorem~\ref{thm:MainN}, we will additionally assume that the buyer is additive across the items, i.e., for any subset of items $S \subseteq [n]$, ${x}(S) = \sum_{i\in S} x_i$. Given a possibly randomized allocation of items $\vec{q} \in [0,1]^n$, we use ${x}(\vec{q})$ to denote the buyer's expected value for the realized set of items.  For an additive buyer, ${x}(\vec{q}) = \sum_{i=1}^{n} x_i \cdot q_i$. For a monotone buyer, $x(\vec{q}) = \sum_{S \subseteq [n]} x(S) \cdot \Pr(\vec{q},S)$, where $\Pr(\vec{q}, S) = \prod_{j \in S} \vec{q}_j \prod_{j \not \in S} (1-\vec{q}_j)$ is the probability that the buyer gets exactly the items in set $S$ from randomized allocation $\vec{q}$. Let $\Lambda = \{\vec{q}_1, \vec{q}_2, \dots, \vec{q}_k, \dots\}$ be a multi-set of allocations (of possibly unbounded size). We denote by $\lottery(\Lambda) \in [0,1]^n$ (read ``lottery") the expected allocation that results from being allocated every $\vec{q}_i \in \Lambda$ independently and at once, i.e., for all $j \in [n]$, $\lottery(\Lambda)_j = 1 - \prod_{i=1} (1-q_{ij})$.\footnote{The careful reader might wonder what would happen if instead of buying all their menu options at once, the buyer was allowed to do so \emph{adaptively} (as opposed to the model presented here which is non-adaptive). We discuss this in Appendix~\ref{sec:adapt}, but the main takeaway is that our upper bounds also hold for adaptive buyers.}


A mechanism $\mech = (p, q)$ is defined by a pair of functions $p:\mathcal{X} \rightarrow \mathcal{R}_{\geq 0}$, $q: \mathcal{X} \rightarrow [0,1]^{n}$ known as the pricing and allocation functions, respectively. For a fixed integer $k$ we say that a mechanism $\mech$ is \emph{buy}-$k$ incentive-compatible if for every valuation $\vec{x} \in \mathcal{X}$ it is in the buyer's best interest to purchase a single option from the mechanism rather than any combination of up to $k$ menu options. In other words, there exists some $(p,\vec{q}) \in \mech$ such that ${x}(\vec{q}({x})) - p({x}) \geq {x}(\lottery(\Lambda)) - \sum_{i=1} p(\vec{q}_i) $ for \emph{any} possible multi-set of menu options $\Lambda$ of size at most $k$. Thus, setting $k=1$ recovers the standard notion of (buy-one) incentive-compatible mechanisms, and as $k \rightarrow \infty$ it recovers the existing definition of buy-many incentive-compatible mechanisms.





\subsection{Benchmarks of Interest} 

We now formally define some of the benchmarks that will be used throughout this paper. For a given (arbitrarily correlated) distribution $\distD$, let 

\begin{itemize}
    \item $\brev{\distD}$ be the revenue of the mechanism which sells the grand bundle for its optimal price. Namely, $\brev{\distD} = \max_{p} p \cdot \Pr_{\vec{x}\sim \distD}(\sum x_i \geq p)$,
    \item $\rev{\distD}$ be the revenue of the optimal buy-one incentive-compatible mechanism,
    \item $\rev{\distD, \mech}$ be the revenue of mechanism $\mech$ when the buyer is allowed to buy up to $1$ menu entry from $\mech$,
    \item $\bkrev{\distD}{k}$ be the revenue of the optimal buy-$k$ incentive-compatible mechanism, 
    \item $\bkrev{\distD, \mech}{k}$ be the revenue of (not necessarily buy-$k$ incentive-compatible) mechanism $\mech$ when the buyer is allowed to buy up to $k$ menu entries from $\mech$, 
    \item $\bmrev{\distD}$ be the revenue of the optimal buy-many incentive-compatible mechanism. 
\end{itemize}

\begin{proof}[Proof of Claim:\ref{cl:bmrev}]
If a mechanism $\mech$ is buy-$k$ incentive-compatible for some $k$, it is also buy-$k'$ incentive-compatible for all $k' \leq k$. This follows from the fact that the buyer can always buy the empty lottery $k-k'$ times and the mechanism must guard against such deviations. Thus, the best buy-$k$ mechanism can perform no better than the best buy-$(k-1)$ mechanism, proving the claim. The last inequality follows from the fact that bundling is a buy-many incentive compatible mechanism.  
\end{proof}

\subsection{Menu Gaps: An Intermediary Measure}

We now present the definition of $\gaptwo{k}{i}{X}{Q}$ and $\MenuGapTwo{k}{X}{Q}$, the quantities that will serve as intermediaries in proving Theorems~\ref{thm:MainN},~\ref{thm:GeneralVal}. 

\begin{definition}
\label{def:main}
Let $X = \{{x}_i\}_{i=1}^N$ be a sequence of monotone valuation functions and $Q = \{\vec{q}_i\}_{i=0}^N \in [0,1]^n$ be a sequence of vectors with $\vec{q}_0 = \vec{0}^n.$ Then 

\begin{equation}
\label{def:gengap}
\gaptwo{k}{i}{X}{Q} = \min_{j_1, j_2, \dots, j_k < i} {x}_i(\vec{q}_i) - {x}_i(\lottery(\vec{q}_{j_1}, \vec{q}_{j_2}, \dots, \vec{q}_{j_k})).
\end{equation}
Furthermore, define 
\begin{equation}
\label{def:menugap}
\MenuGapTwo{k}{X}{Q} = \sum_{i=1}^N \gaptwo{k}{i}{X}{Q}/{x}_i([n]),  
\end{equation}
where ${x}_i([n])$ is the valuation of a buyer of type ${x}_i$ for the grand bundle of items.

\end{definition}

In particular, if $X=\{x_i\}_{i=1}^{N}$ are additive valuations, then they can be represented as vectors $X=\{\vec{x}_i\}_{i=1}^{N}\in\mathbb{R}^n_{\ge0}$. One can check that the following definition is a special case of Definition \ref{def:main}.

\begin{definition}
\label{def:main_additive}
Let $X = \{\vec{x}_i\}_{i=1}^N \in \mathbb{R}^{n}_{\geq 0}, Q = \{\vec{q}_i\}_{i=0}^N \in [0,1]^n$ be sequences of vectors with $\vec{q}_0 = \vec{0}^n.$ Then 

\begin{equation}
\label{def:addgap}
\gaptwo{k}{i}{X}{Q} = \min_{j_1, j_2, \dots, j_k < i} \vec{x}_i \cdot (\vec{q}_i - \lottery(\vec{q}_{j_1}, \vec{q}_{j_2}, \dots, \vec{q}_{j_k})).
\end{equation}
Furthermore, define 
\begin{equation}
\label{def:addmenugap}
\MenuGapTwo{k}{X}{Q} = \sum_{i=1}^N \gaptwo{k}{i}{X}{Q}/||\vec{x_i}||_1,  
\end{equation}

\end{definition}

These measures are generalizations of similar notions introduced in~\cite{HartN17} and further developed by~\cite{PSW22}. For the case where $k=1$, we exactly recover these earlier definitions. In Definition~\ref{def:main}, it is useful to think of the first sequence of vectors $X$ as possible valuation vectors and the sequence of vectors $Q$ as possible allocation vectors of a mechanism, with the built-in option of not participating. Thus, one way to interpret Equation~\ref{def:gengap} is to think of $p_i = \gaptwo{i}{k}{X}{Q}$ as the largest price a seller can post on menu entry $(p_i, \vec{q}_i)$ so that a buyer with valuation ${x}_i$ will prefer that single menu entry to any subset of at most $k$ ``previous" options for \emph{free}.

\subsection{Some Useful Properties}

We prove a simple, useful property of the $\lottery(\Lambda)$ function. Namely, that if $\Lambda = \{\vec{q_1}, \vec{q}_2, \dots, \vec{q}_k \}$, then $\lottery(\Lambda)$ dominates the vector consisting of coordinate-wise max entries of the vectors in $\Lambda$. 

\begin{claim}
\label{cl:lotprod}
If $\Lambda = \{\vec{q}_1, \vec{q}_2, \dots, \vec{q}_k\}, \vec{q}_i \in [0,1]^n$ for all $i$, then $\lottery (\Lambda)_j \geq \max_{i \in [k]} \{q_{ij}\}$ for all $j \in [n].$
\end{claim}

\begin{proof} Fix $j \in [n]$. Assume wlog $q_{1j} = \max_{i \in k} q_{ij}$. It is easy to see that $(1-q_{1j}) \cdot (1-\prod_{j=2}^{k}(1-q_{ij})) \geq 0$ since each term on the left is non-negative. Expanding it we get that $1-q_{1j}-\prod_{j=1}^{k}(1-q_{ij})\geq 0$. Rewriting gives $\lottery(\Lambda)_j \geq q_{1j}$. 
\end{proof}

We also prove a simple, useful property of the $\MenuGapTwo{k}{X}{Q}$ function. Namely, that it is without loss of generality to remove points whose contributions to the sum are non-positive. 

\begin{claim} 
\label{cl:subseq}
Let $X,Q$ be sequences as defined in Definition~\ref{def:main}, and $X' \subseteq X, Q' \subseteq Q$ be the sub-sequences that result from removing any pair of points $(\vec{x}_i, \vec{q}_i)$ whose $\gaptwo{k}{i}{X}{Q} \leq 0$. Then $$\MenuGapTwo{k}{X}{Q} \leq \MenuGapTwo{k}{X'}{Q'}.$$ 
\end{claim}

\begin{proof}
Consider the earliest integer $i$ such that $\gaptwo{k}{i}{X}{Q} \leq 0$. Since it is non-positive, removing $({x}_i, \vec{q}_i)$ from $(X,Q)$ will weakly increase the sum of the gaps up to $i$. Moreover, if $\vec{q}_i$ was helping set the gap for some later $({x}_j, \vec{q}_j)$, then $\gaptwo{k}{j}{X'}{Q'} \geq \gaptwo{k}{j}{X}{Q}$ since by removing $\vec{q}_i$ we are reducing the number of earlier points to compare to. Therefore, removing any point with negative gap can only weakly increase the menu gap of the resulting subsequence.
\end{proof}

\section{Bundling Approximates the Optimal Buy-N Revenue}
\label{sec:main}

As highlighted in the introduction of the paper, the first part of the proofs of Theorems~\ref{thm:MainN},~\ref{thm:GeneralVal} will be via the surrogate quantity, $\MenuGapTwo{k}{X}{Q}$. We will first show that for any distribution over $n$ items $\distD$, there exists two sequences of points $(X, Q)$ such that $\MenuGapTwo{k}{X}{Q}$ upper bounds the ratio between the revenue-optimal buy-$k$ mechanism for $\distD$ and the revenue from bundling, up to a $O(k)$ factor: 

\begin{lemma}
\label{lem:LBnk}
For any distribution of monotone valuations $\distD$ over $n$ items and any buy-$k$ incentive compatible mechanism $\mech$ for $\distD$, there exists a sequence of valuations $X = \{{x_i}\}_{i=1} \subseteq \mathcal{X}$, and a sequence of allocations $Q = \{\vec{q}_i\}_{i=0} \subseteq \mech$ (starting with $\vec{q}_0 = (0,\dots,0)$) such that 
$$ \MenuGapTwo{k}{X}{Q} \geq \frac{\bkrev{\distD, \mech}{k}}{9k \cdot \brev{\distD}}.$$
\end{lemma}

Next we will show that this quantity itself is upper bounded for \emph{all} pairs of sequences $(X,Q)$. In particular, when the buyer has additive valuations, we show that the quantity is upper bounded by $n$. The proof of Theorem \ref{thm:MainN} then follows directly.

\begin{lemma}
\label{lem:UBadditive}
For all sequences $X, Q$ as defined in Definition~\ref{def:main_additive}, coming from an additive valuation function, it holds that $\MenuGapTwo{n}{X}{Q} \leq n$.
\end{lemma}

\thmMainN*
\begin{proof}[Proof of Theorem~\ref{thm:MainN}]
Follows directly from Lemma~\ref{lem:LBnk} by choosing $\mech$ to be the revenue-optimal buy-$k$ incentive-compatible mechanism and Lemma~\ref{lem:UBadditive} (setting $k = n$). 
\end{proof}

The proof of Theorem~\ref{thm:GeneralVal} is similar, but only a weaker version of Lemma~\ref{lem:UBadditive} can be proved:

\thmGeneralVal*

\begin{lemma}
\label{lem:UBnn}
For all sequences $X, Q$ as defined in Definition~\ref{def:main}, coming from a monotone valuation function, it holds that $\MenuGapTwo{n}{X}{Q} \leq n \cdot 2^n$.
\end{lemma}

\begin{proof}[Proof of Theorem~\ref{thm:GeneralVal}]
Follows directly from Lemma~\ref{lem:LBnk} by choosing $\mech$ to be the revenue-optimal buy-$k$ incentive-compatible mechanism and Lemma~\ref{lem:UBnn} (setting $k = n$). 
\end{proof}

Note that in Lemma~\ref{lem:LBnk}, the number of times the buyer can interact with the mechanism, $k$, may be different than the number of items $n$ for sale. However, we are only able to prove Lemma~\ref{lem:UBnn} for the case where $k \geq n$. We hope future work can address the question of what happens when $k<n$.

\subsection{Menu Gap Approximately Upper Bounds Revenue Gap: Proof of Lemma~\ref{lem:LBnk}}

The proof of Lemma~\ref{lem:LBnk} is split into two parts. In the first part, we will any buy-$k$ incentive-compatible mechanism for $\distD$ and massage it down to a sub-menu of interest whose revenue remains close to the original one. The sub-menu itself may not be buy-$k$ incentive-compatible. However, the key to approximately preserving the revenue will be in just removing entries from the original mechanism and not modifying existing ones. In the second part, we will show how to use an appropriate sub-menu in order to construct the desired sequence of points. The proof of Lemma~\ref{lem:LBnk} is inspired on a similar construction of~\cite{PSW22}.

\subsubsection{Finding a Sub-menu of Interest}

Let $\mech = \{(p_i,\vec{q}_i)\}_{i=1}$ be a buy-$k$ incentive compatible mechanism for distribution $\distD$, where $(p_i,\vec{q}_i)$ denotes the price and expected allocation of the $i$-th option of the menu. 

\begin{claim}
\label{cl:discn}
Let $\mech$ be a buy-$k$ incentive-compatible mechanism for $\distD$, $\mech_c \subseteq \mech$ be the sub-menu of $\mech$ that only offers options of price at least $c$. Then $\bkrev{\distD, \mech_c}{k} \geq \bkrev{\distD, \mech}{k}-c$. 
\end{claim}

\begin{proof}
If a buyer with valuation $x$ chose a menu entry $(p,\vec{q})$ from the original menu $\mech$ with $p \geq c$, they will purchase the same menu entry in $\mech_c$ since $(p,\vec{q})$ was utility-maximizing and no new menu entries were introduced. If $p < c$, it is possible that the buyer would purchase some other option $(p', \vec{q}')$ (or combination of options). Regardless, the loss in revenue from that buyer is bounded by $c f(x)$. Let $S_c$ be the set of valuation vectors that preferred a menu entry in $\mech$ priced at $p < c$. Then the total loss in revenue is at most $c \sum_{x \in S_c} f(x) \leq c$.  \end{proof}

\begin{claim}
\label{cl:splitn}
Let $\mech$ be a mechanism, and let $\mech_1, \mech_2 \subseteq \mech$ be sub-menus of $\mech$ defined as follows: 

\begin{itemize}
    \item $\mech_1$ has all options whose price $p_i \in \cup_{i=0}^{\infty} [c \cdot (2k)^{2i}, c\cdot (2k)^{2i+1})$. 
    \item $\mech_2$ has all options whose price $p_i \in \cup_{i=0}^{\infty} [c \cdot (2k)^{2i+1}, c \cdot (2k)^{2i+2})$. 
\end{itemize}
Then $\max_{i=1,2} \bkrev{\distD, \mech_i}{k} \geq \bkrev{\distD, \mech}{k}/2$.
\end{claim}

\begin{proof}
Because $\mech_1 \cup \mech_2 = \mech$, observe that 
$$\bkrev{\distD, \mech}{k} \leq \bkrev{\distD, \mech_1}{k} + \bkrev{\distD, \mech_2}{k}.$$ This is because any buyer with valuation $x$ who purchases an option from $\mech_1$ when presented the menu $\mech$ will buy the same option when only presented $\mech_1$. By a simple averaging argument, the better of the two menus must get revenue at least half of the revenue of the original menu.
\end{proof}

\begin{lemma}
\label{lem:submenun}
There exists a menu $\mech$ such that 

\begin{itemize}
    \item All prices are at least $c$.
    \item All prices belong to the set of intervals $\bigcup_{i=0}^{\infty} [c \cdot (2k)^{2i+a}, c \cdot (2k)^{2i+a+1})$ for an $a \in \{0,1\}$.
    \item $\bkrev{\distD, \mech}{k} \geq \frac{\bkrev{\distD}{k}-c}{2}$.  
\end{itemize}

\end{lemma}

\begin{proof}
Take the initial buy-$k$ incentive-compatible menu $\mech$, apply Claim~\ref{cl:discn} to obtain a menu $\mech'$ that satisfies the first bullet point. Take the menu $\mech'$ and apply Claim~\ref{cl:splitn} to obtain a menu $\mech''$ that immediately satisfies the first and second bullet points. Finally, due to the revenue guarantees of Claims~\ref{cl:discn},~\ref{cl:splitn}, we have that $\bkrev{\distD, \mech''}{k} \geq \frac{\bkrev{\distD, \mech'}{k}}{2} \geq \frac{\bkrev{\distD, \mech}{k}-c}{2}$. 
\end{proof}

\subsubsection{Construction of the Sequences $X, Q$}

In this subsection we will show how to use the sub-menu found in the previous subsection to construct the sequences $(X,Q)$ of interest who would witness 
$$\MenuGapTwo{k}{X}{Q} \geq O(\bkrev{\distD, \mech}{k}/\brev{\distD}).$$ Consider the menu $\mech''$ from Lemma~\ref{lem:submenun}. Let $\mathcal{B}_j \subseteq \mech''$ be the sub-menu that has all menu entries priced in $[c \cdot (2k)^{2j+a}, c \cdot (2k)^{2j+a+1})$ for the same $a \in \{0, 1\}$ from Lemma~\ref{lem:submenun}. We say a valuation ${x} \in \mathcal{B}_j$ if the menu option $(p(x),\vec{q}(x)) \in \mathcal{B}_j$. Let ${x}_j$ be the valuation on $\mathcal{B}_j$ such that ${x}_j([n]) \leq (1+\delta) {x}([n]) \forall x \in \mathcal{B}_j$. We call valuation $x_j$ the \emph{representative of bin} $\mathcal{B}_j$. 

\begin{claim}
\label{cl:bundleprobn}
$\Pr(x \in \mathcal{B}_j) \leq \frac{\brev{\distD}(1+\delta)}{{x}_j([n])}$.
\end{claim}

\begin{proof}
Consider the mechanism that sells the grand bundle at price ${x}_j([n])/(1+\delta)$. Since any valuation on $\mathcal{B}_j$ has value at least that much for the grand bundle, the revenue of this menu is at least $\Pr({x} \in \mathcal{B}_j) x_j([n])/(1+\delta)$. But this is a grand bundling menu and is thus its revenue is at most $\brev{\distD}$. 
\end{proof}

Let $(X,Q)$ be the sequence defined by the choice of ${x}_j$ and their respective allocations in $\mech$, $\vec{q}_j$.

\begin{claim}
\label{cl:gapn}
$\frac{\gaptwo{k}{j}{X}{Q}}{{x}_j([n])} \geq \frac{p_j}{2 \cdot {x}_j([n])} \geq \frac{ \Pr(x \in \mathcal{B}_j) p_j}{2 \cdot \brev{\distD} (1+\delta)}$. 
\end{claim}

\begin{proof}
Because the initial mechanism $\mech^*$ was buy-$k$ incentive-compatible, we know that for any previous set of $k$ options $\vec{q}_{j_1},\dots, \vec{q}_{j_k}$,  

$${x}_j(\vec{q}_j) - p_j \geq {x}_j(\lottery(\vec{q}_{j_1}, \dots,\vec{q}_{j_k})) -\sum_{i=1}^k {p}_{j_i}.$$
We can rewrite this as 

$$ \frac{\gaptwo{k}{j}{X}{Q}}{{{x}_j([n])}} \geq \frac{p_j - \sum_{i=1}^k p_{j_i}}{{x}_j([n])}.$$ 
Recall by our choice of points and the fact that $j_i < j$, $p_{j} \geq 2k \cdot p_{j_i}$. Therefore, the right hand is at least  $\frac{p_j}{2 {x}_j([n])}.$ The second inequality comes from Claim~\ref{cl:bundleprobn}. 
\end{proof}

We are now ready to prove Lemma~\ref{lem:LBnk}.

\begin{proof}[Proof of Lemma~\ref{lem:LBnk}]

Let us first observe that 

\begin{equation}
\label{eq:pricegapn}
 \bkrev{\distD, \mech''}{k} = \sum_j \sum_{{x} \in \mathcal{B}_j} p({x}) f({x}) \leq \sum_j \Pr({x} \in \mathcal{B}_j) 2k p_j.
\end{equation}
 
Recall that the price any valuations $x\in \mathcal{B}_j$, its price $p(x)$ is no greater than $2k p_j$. Therefore, the inequality follows. Moreover, from Claim~\ref{cl:gapn} we get that 

\begin{equation} 
\label{eq:gapnlb}
\MenuGapTwo{k}{X}{Q} = \sum_j \frac{\gaptwo{k}{j}{X}{Q}}{x_j([n])} \geq \sum_j \frac{ \Pr(x \in \mathcal{B}_j)p_j}{2 \cdot \brev{\distD} (1+\delta)}.
\end{equation}
Applying Eq.~\ref{eq:pricegapn} together with Lemma~\ref{lem:submenun} we get that
\begin{equation}
\label{eq:gapnub}
\sum_j \Pr(x \in \mathcal{B}_j) p_j \geq \frac{\bkrev{\distD, \mech''}{k}}{2k} \geq \frac{(\bkrev{\distD, \mech}{k}-c)}{4k}. 
\end{equation}
Putting Eqs.~\ref{eq:gapnlb},~\ref{eq:gapnub} we get that
\begin{equation} 
\label{eq:finaln}
\MenuGapTwo{k}{X}{Q} \geq \frac{ (\bkrev{\distD, \mech}{k}-c)}{8k\cdot \brev{\distD}(1+\delta)}.
\end{equation}
Let $c = \bkrev{\distD, \mech}{k}/100, \delta = 1/100$ in Equation~\ref{eq:finaln}. Therefore, 

$$ \frac{99 \cdot \bkrev{\distD, \mech}{k}}{101\cdot 8k \cdot \brev{\distD}} \geq \frac{\bkrev{\distD,\mech}{k}}{9 \cdot k \cdot \brev{\distD}}.$$ 
\end{proof}

\subsection{Menu Gap is Finite when $k \geq n$: Proof of Lemmas~\ref{lem:UBadditive},~\ref{lem:UBnn}}

The proof of Lemmas~\ref{lem:UBadditive},~\ref{lem:UBnn} will be similar. We introduce some common notation to both Lemmas before proving each of them individually. Given a sequence of points $Q$, let $Q_i$  be the sequence truncated at the $i$-th point, that is to say $Q_i = \{\vec{q}_0, \vec{q}_1, \dots, \vec{q}_i \}$. Let $\vec{m}_i = (\max_{\vec{q}_{i} \in Q_i} \{\vec{q}_{i,1}\}, \dots, \max_{\vec{q}_{i} \in Q_i} \{\vec{q}_{i,n}\})$ be the $n$-dimensional vector whose entries are the largest coordinates among the points in $Q_i$. 

\subsubsection{Proof Lemma~\ref{lem:UBadditive}}

In this subsection we abuse the fact that for additive buyers, we can think of their valuation function $x$ as an $n$-dimensional vector $\vec{x} = (x_1, \dots, x_n)$ and $x(\vec{q}) = \vec{x}\cdot \vec{q}$. 

\begin{claim}
\label{cl:maxgapn}
For any $i$, $\gaptwo{n}{i}{X}{Q}/||x_i||_1 \leq \sum_{d=1}^n \max \{\vec{q}_{i,d}-\vec{m}_{i-1,d}, 0\}.$
\end{claim}

\begin{proof}
Let $\vec{x}_i, \vec{q}_i$ be given. Since gap is defined to be the minimum over all pairs of previously placed points, we can just upper bound it by witnessing its value with earlier points. Let $i^*_1, i^*_2,\dots, i^*_n$ be the indices such that: 

\begin{itemize}
    \item $i^*_1, i^*_2, \dots, i^*_n < i,$
    \item $\vec{q}_{i^*_d,d} = \vec{m}_{i-1,d} \forall d \in [n].$
\end{itemize}
That is to say, $\{i^*_d\}_{d=1}^{n}$ are the indices of the points that witness that $\vec{m}_{i-1}$ is indeed the coordinate-wise max of all points in $Q_{i-1}$. Then 
$$\gaptwo{n}{i}{X}{Q}/||\vec{x}_i||_1 \leq \frac{\vec{x}_i}{||\vec{x_i}||_1} \cdot (\vec{q}_i - \lottery(\vec{q}_{i^*_1}, \vec{q}_{i^*_2},\dots,\vec{q}_{i^*_n})).$$
Recall that $$\lottery(\vec{q}_{i^*_1}, \vec{q}_{i^*_2},\dots,\vec{q}_{i^*_n})_d \geq \max\{\vec{q}_{i^*_1}, \vec{q}_{i^*_2},\dots,\vec{q}_{i^*_n}\}_d \geq \vec{m}_{i-1,d}.$$
Therefore, $$ \vec{q}_{i,d} - \lottery(\vec{q}_{i^*_1}, \vec{q}_{i^*_2},\dots, \vec{q}_{i^*_n})_d \leq \vec{q}_{i,d} - \vec{m}_{i-1,d},$$
for all $d \in [n]$. Therefore, for any choice of $\vec{x}_i$, it will be true that 
\begin{align*}
\gaptwo{n}{i}{X}{Q}/||\vec{x}_i||_1 & \leq \frac{\vec{x}_i}{||\vec{x}_i||_1} \cdot (\vec{q}_i - \lottery(\vec{q}_{i^*_1}, \vec{q}_{i^*_2},\dots,\vec{q}_{i^*_n})) \\ 
& = \sum_{d=1}^n \frac{\vec{x}_{i,d}}{||\vec{x}_i||_1} (\vec{q}_{i,d} - \lottery(\vec{q}_{i^*_1},\vec{q}_{i^*_2},\dots,\vec{q}_{i^*_n})_d) \\ 
& \leq \sum_{d=1}^{n} \frac{\vec{x}_{i,d}}{||\vec{x}_i||_1} (\vec{q}_{i,d} - \vec{m}_{i-1,d}) \\
& \leq \sum_{d=1}^{n} \max\{0, \vec{q}_{i,d} - \vec{m}_{i-1,d}\} 
\end{align*}
This proves the claim (Naturally, $\vec{x}_{i,d} \leq ||\vec{x}_i||_1$ for all $d$).
\end{proof}

\begin{proof}[Proof of Lemma~\ref{lem:UBadditive}]
For any pair of sequences $(X,Q)$ coming from an additive valuation function,  
\begin{align*}
\MenuGapTwo{n}{X}{Q} & = \sum_{i=1}^N \gaptwo{n}{i}{X}{Q}/||x_i||_1 \\ 
& \leq \sum_{d=1}^{n} \sum_{i=1}^N (\max \{\vec{q}_{i,d}-\vec{m}_{i-1,d}, 0\}) \\
& \leq \sum_{d=1}^{n} \sum_{i=1}^N (\max \{\vec{m}_{i,d}-\vec{m}_{i-1,d}, 0\}) \\ 
& \leq \sum_{d=1}^{n} \sum_{i=1}^N (\vec{m}_{i,d}-\vec{m}_{i-1,d}) \\ 
& \leq n.
\end{align*}

The first inequality follows from Claim~\ref{cl:maxgapn}. For the second inequality, first note that for any fixed $d$, by definition $\vec{q}_{i,d} \leq \vec{m}_{i,d}$, with equality only if $\vec{q}_{i,d} \geq \vec{q}_{i',d}$ for all $i' < i$. But note also that by definition $\vec{m}_{i,d}-\vec{m}_{i-1,d} \geq 0$. Therefore $\max \{\vec{m}_{i,d}-\vec{m}_{i-1,d}, 0 \} \leq \vec{m}_{i,d}-\vec{m}_{i-1,d}$. The last inequality follows from observing that the final sum across each coordinate telescopes. Since $\vec{q}_{i,d} \leq 1$, the sum is at most $1$ per coordinate. 
\end{proof}

\begin{observation}
Setting both sequences $(X,Q)$ equal to the standard basis of $\mathbb{R}^n$ shows that Lemma~\ref{lem:UBadditive} is tight.
\end{observation}

\subsubsection{Proof of Lemma~\ref{lem:UBnn}}
In this subsection we assume the valuation function is monotone, i.e. $v(S) \geq v(T)$ whenever $T \subseteq S$. We will show that for such valuation classes, $\MenuGapTwo{n}{X}{Q}\le 2^{n}\cdot n$ for all sequences $(X,Q)$. To prove this, we'll make use of the following lemma.

\begin{claim}
For any monotone valuation $v$ and two vectors $p,q\in[0,1]^n$ with $p\le q$ coordinate-wise (i.e., $p_i\le q_i$ for all $i\in[n]$), we have $v(p)\le v(q)$.\label{cl:monotone}
\end{claim}
\begin{proof}
To simplify the proof, observe that we only need to consider the case where $p$ and $q$ differ in exactly one coordinate (i.e., $p_i=q_i$ for $i\in[n-1]$ and $p_n<q_n$). We can then apply this result (up to) $n$ times in order to conclude our desired result. 

Since we can partition the powerset of $[n]$ into sets which contain $n$ and sets which don't contain $n$, we can rewrite $v(\vec{p})$ as:
\begin{align}
    v(\vec{p})=\sum_{A\subseteq [n]}\Pr(\vec{p},A)\cdot v(A)=\sum_{B\subseteq [n-1]}\Pr(\vec{p},B\cup\{n\})\cdot v(B\cup\{n\})+\sum_{B\subseteq[n-1]}\Pr(\vec{p},B)\cdot v(B)
\end{align}
This implies that $v(\vec{q})-v(\vec{p})=A+B$, where $A,B$ are defined as:
\begin{align}
    A&=\sum_{B\subseteq [n-1]}[\Pr(\vec{q},B\cup\{n\})-\Pr(\vec{p},B\cup\{n\})]\cdot v(B\cup\{n\})\\
    B&=\sum_{B\subseteq[n-1]}[\Pr(\vec{q},B)-\Pr(\vec{p},B)]\cdot v(B)
\end{align}
Our goal is to show $A+B\ge 0$. To do this, let's first find explicit expressions of $[\Pr(\vec{q},B\cup\{n\})-\Pr(\vec{p},B\cup\{n\})]$ and $[\Pr(\vec{q},B)-\Pr(\vec{p},B)]$. We have:
\begin{align}
    [\Pr(\vec{q},B\cup\{n\})-\Pr(\vec{p},B\cup\{n\})]=(q_n-p_n)\prod_{i\in B}q_i\cdot\prod_{i\not\in B\cup\{n\}}(1-q_i)
\end{align}
and
\begin{align}
    [\Pr(\vec{q},B)-\Pr(\vec{p},B)]=[(1-q_n)-(1-p_n)]\prod_{i\in B}q_i\cdot\prod_{i\not\in B\cup\{n\}}(1-q_i)    
\end{align}
by direct calculation (and noting that $p_i=q_i$ for all $i\in[n-1]$). Thus, we observe that the two probabilities are additive complements of each other. Hence, if we denote $p_B=[\Pr(\vec{q},B\cup\{n\})-\Pr(\vec{p},B\cup\{n\})]$, we can write
\begin{align}
    v(\vec{q})-v(\vec{p})=\sum_{B\subseteq[n-1]}p_B\cdot[v(B\cup\{n\})-v(B)].
\end{align}
But we know the right hand side is at least $0$ since $v(B\cup\{n\})\ge v(B)$ by monotonicity. Hence, we can conlude $v(\vec{q})-v(\vec{p})\ge 0\implies v(\vec{p})\le v(\vec{q})$.
\end{proof}

We are now ready to prove Lemma~\ref{lem:UBnn}.

\begin{proof}[Proof of Lemma~\ref{lem:UBnn}]
Recall that, for all $i$, we can choose $i_1^*,\ldots,i_n^*<i$ such that $\vec{q}_{i^*_d,d}=\vec{m}_{i-1,d}$ for all $d\in[n]$. Thus, by monotonicity and Claim~\ref{cl:monotone}, we have $\max_{j_1,\ldots,j_n<i}v(\lottery(q_{j_1},\ldots,q_{j_k}))\ge v(\vec{m}_{i-1})$. Additionally, by definition of $\vec{m}_i$ and monotonicity (combined with Claim \ref{cl:monotone}), we have $v(\vec{q}_i)\le v(\vec{m}_i)$. Consequently, we can write
\begin{align*}
\frac{\gaptwo{k}{i}{X}{Q}}{v_i([n])}&=\min_{j_1,\ldots,j_k<i}\frac{v_i(\vec{q}_i)-v_i(\lottery(\vec{q}_{j_1},\ldots,\vec{q}_{j_k}))}{v_i([n])}\le\frac{v_i(\vec{m}_i)-v_i(\vec{m}_{i-1})}{v_i([n])}
\end{align*}
since $v_i$ are assumed to be monotone. Writing out the definition of $v_i$, we know
\begin{align*}
\frac{v_i(\vec{m}_i)-v_i(\vec{m}_{i-1})}{v_i([n])}&= \sum_{A\in2^{[n]}}\frac{v_i(A)}{v_i([n])}\cdot\left[\Pr(\vec{m}_{i}, A)-\Pr(\vec{m}_{i-1},A)\right]\\
&\le \sum_{A\in2^{[n]}}|\Pr(\vec{m}_{i},A)-\Pr(\vec{m}_{i-1},A)|
\end{align*}
since $v_i(A)\le v_i([n])$ by monotonicity. Combining this with above yields
\begin{align*}
    \MenuGapTwo{n}{X}{Q}&=\sum_{i=1}^{N}\frac{\gaptwo{n}{i}{X}{Q}}{v_i([n])}\\
    &\le \sum_{A\in2^{[n]}} \sum_{i=1}^{N} |\Pr(\vec{m}_{i},A)-\Pr(\vec{m}_{i-1},A)|
\end{align*}
It remains to show
$\sum_{i=1}^{N} |\Pr(\vec{m}_{i},A)-\Pr(\vec{m}_{i-1},A)|\le n$
for all $A\in2^{[n]}$. It would then follow directly that $\MenuGapTwo{n}{X}{Q}\le 2^{n}\cdot n$. Fix a set $A$ and consider each term $|\Pr(\vec{m}_i,A)-\Pr(\vec{m}_{i-1},A)|$ individually. 

Let $\vec{m}_i=\vec{m}^{(0)},\vec{m}^{(1)},\ldots, \vec{m}^{(n)}=\vec{m}_{i-1}$ be a sequence of vectors where $\vec{m}^{(j)}$ matches $\vec{m}_{i}$ on the first $j$ coordinates and matches $\vec{m}_{i-1}$ on the last $n-j$ coordinates. By the triangle inequality,
$$|\Pr(\vec{m}_i,A)-\Pr(\vec{m}_{i-1},A)|\le \sum_{j=1}^{n}|\Pr(\vec{m}^{(j)},A)-\Pr(\vec{m}^{(j-1)},A)|.$$
But by definition of $\Pr(\vec{q},S)$, we know
$$|\Pr(\vec{m}^{(j)},A)-\Pr(\vec{m}^{(j-1)},A)|= \left|\prod_{\ell \in A} \vec{m}^{(j)}_\ell \prod_{\ell\not \in A} (1-\vec{m}^{(j)}_\ell)-\prod_{\ell \in A} \vec{m}^{(j)}_\ell \prod_{\ell\not \in A} (1-\vec{m}^{(j)}_\ell)\right|\le \vec{m}_{i,j}-\vec{m}_{i-1,j}$$ since $\vec{m}^{(j)}$ and $\vec{m}^{(j-1)}$ differ only in the $j^{th}$ coordinate and the remaining probabilities are at most one. Hence, combining with the above, we have
$$|\Pr(\vec{m}_i,A)-\Pr(\vec{m}_{i-1},A)|\le \sum_{j=1}^{n}\vec{m}_{i,j}-\vec{m}_{i-1,j}.$$
Summing over $i\in[N]$, we see that the sum telescopes:
$$\sum_{i=1}^{N}|\Pr(\vec{m}_i,A)-\Pr(\vec{m}_{i-1},A)|\le \sum_{j=1}^{n}\sum_{i=1}^{N}\vec{m}_{i,j}-\vec{m}_{i-1,j}\le \sum_{j=1}^{n}\vec{m}_{N,j}\le n,$$
since $\vec{m}_{N,j}\le 1$ for all $j\in[n]$.
\end{proof}
\section{Few Tickets Do Not Suffice: Proof of Theorem~\ref{thm:mainLB}}
\label{sec:few}

In this section we show that if $k \leq n^{1/2-\varepsilon}$, then there exists a distribution $\distD$ over $n$ items for which there is an exponential gap in $n$ (up to $\textsc{poly}(n)$) between $\brev{\distD}$ and $\bkrev{\distD}{k}$, the revenue attained by the optimal buy-$k$ mechanism. 

\thmMainLB*

The proof of Theorem~\ref{thm:mainLB} will be broken
down in three steps. Firstly, we will describe the pair
of sequences $(X^L,Q^L)$ that we use
(Subsection~\ref{sec:instance}). The construction will
make use of a combinatorial Lemma about cover-free sets
from~\cite{KautzSingleton64}. Next, we will show that
for that instance, $\MenuGapTwo{k}{X^L}{Q^L} \geq
\frac{\exp\left(\Omega(n^\varepsilon)\right)}{2n^2}$
when $k \leq n^{1/2-\varepsilon}$
(Lemma~\ref{cl:expmenugap}). In the final step, we will
show how to construct a distribution $\distD$ such that
$\bkrev{\distD}{k}/\brev{\distD} \geq
\MenuGapTwo{k}{X^L}{Q^L}$ (Lemma~\ref{lem:hnconverse}).
The proof of Lemma~\ref{lem:hnconverse} will use ideas from~\cite{HartN17}. 

Before we delve into the proof of Theorem~\ref{thm:mainLB}, we analyze its implications for mechanisms with polynomial menu size. We invoke the following Corollary from~\cite{HartN17}. 

\begin{corollary}[Restated from~\cite{HartN17}]
\label{cor:bundlebound}
Consider any mechanism $\mathcal{M}$ with menu size $M$, then for any distribution $\distD \in \mathbb{R}^{n}$ 
 
\begin{equation*}
    M \cdot \brev{\distD} \geq \rev{\distD, \mathcal{M}}.
\end{equation*}
\end{corollary}

This Corollary, combined with Theorem~\ref{thm:mainLB} imply the following Corollary. 

\begin{corollary}
\label{cor:polymenubound} 
Let $\mathcal{M}$ be a buy-$k$ mechanism with menu size $M=\textsc{poly}(n)$, then there exists a distribution $\distD \in \mathbb{R}^{n}$ such that for any single, additive buyer

\begin{equation*}
    \frac{\bkrev{\distD}{k}}{\rev{\distD, \mathcal{M}}} \geq \frac{\exp\left(\Omega (n^{\varepsilon})\right)}{\textsc{poly}(n)}.
\end{equation*}
\end{corollary}

In other words, Corollary~\ref{cor:polymenubound} rules out all polynomial-sized mechanisms $\mathcal{M}$ as candidates for good approximations to $\bkrev{\distD}{k}$ for the case $k \leq n^{1/2-\varepsilon}$. 

\subsection{Proof of Theorem~\ref{thm:mainLB}} 

Throughout this Subsection we again abuse the fact that for additive valuation functions $x$, we can think of the valuation as an $n$-dimensional vector $\vec{x} = (x_1, \dots, x_n)$. 

\subsubsection{Part 1: Description of the Instance}
\label{sec:instance}

In order to describe the instance we consider, we first need to introduce the concept of $k$-cover-free families of sets. 

\begin{definition}
A family of sets $\mathcal{F}$ is called $k$-\emph{cover-free} if $A_0 \not \subseteq A_1 \cup A_2 \cup \dots \cup A_k$ holds for all distinct $A_0, A_1, \dots, A_k \in \mathcal{F}$.  
\end{definition}

We will be interested in constructing the largest possible family of sets that is $k$-cover-free. Let $T(n,k)$ denote the maximum cardinality of a $k$-cover-free family of sets $\mathcal{F}$. We use the following bound from~\cite{Furedi96}, attributed there to~\cite{KautzSingleton64}.  

\begin{theorem}[\cite{KautzSingleton64}] For all $n, k$, it holds that 

$$ \Omega \left({\frac{1}{k^2}}\right) \leq \frac{\log(T(n,k))}{n}.$$

In other words, there exists a family of sets $\mathcal{F}_k$ that is $k$-cover-free and $|\mathcal{F}_k| \geq 2^{\Omega\left(\frac{n}{k^2}\right)}$.

\end{theorem}

We will use $k$-cover-free sets to construct pairs of sequences that have large menu gaps. Then, we will take this pair of sequences and show how to obtain a $n$-dimensional distribution whose revenue gap is lower bounded by the menu gap of the underlying pair of sequences. We are now ready to define the instance of interest. Assume $k \leq n^{1/2-\varepsilon}$.

\begin{definition}
\label{def:instance}
Let $\mathcal{F}^L$ be a $k$-cover-free family of sets of maximal size, i.e., such that $|\mathcal{F}^L| = T(n, k) = \exp\left(\Omega(n/k^2)\right)$. Set $\vec{x}^L_i = \vec{q}^L_i = \vec{e}_{A_i} \forall A_i \in \mathcal{F}^L$, where by $\vec{e}_S$ we denote the $n$-dimensional indicator vector for set $S$.
\end{definition}

Observe that unlike other constructions (e.g.,~\cite{HartN17},~\cite{PSW22}) the number of points in each pair of sequences is finite. Thus this instance cannot witness an infinite revenue gap, but we claim it can witness an exponential revenue gap. 

\subsubsection{Part 2: The Instance has Large Menu Gap}

In the next step of the proof of Theorem~\ref{thm:mainLB} we will show that the constructed instance has large menu gap. 

\begin{claim}
\label{cl:expmenugap}
For the instance described in Definition~\ref{def:instance}, it holds that 
\begin{equation*}
    \MenuGapTwo{k}{X^L}{Q^L} \geq  \frac{|\mathcal{F}^L|}{n}.
\end{equation*}
\end{claim}

\begin{proof}

\begin{align*}
        \gaptwo{k}{i}{X^L}{Q^L} & = \min_{j_1, j_2, \dots, j_k \leq i} \frac{\vec{e}_{A_i}}{|\vec{e}_{A_i}|} \cdot \left(\vec{e}_{A_i} - \lottery(\vec{e}_{A_{j_1}},\vec{e}_{A_{j_2}},\dots, \vec{e}_{A_{j_k}})\right) \\
        & \geq \min_{j_1, j_2, \dots, j_k \neq i} \frac{\vec{e}_{A_i}}{|\vec{e}_{A_i}|} \cdot \left(\vec{e}_{A_i} - \lottery(\vec{e}_{A_{j_1}},\vec{e}_{A_{j_2}},\dots, \vec{e}_{A_{j_k}})\right) \\ 
        & \geq \min_{j_1, j_2, \dots, j_k \neq i} \frac{\vec{e}_{A_i}}{|\vec{e}_{A_i}|} \cdot \left(\vec{e}_{A_i} - \vec{e}_{\cup_{\ell=1}^k A_{j_\ell}}\right) \\
        & \geq \min_{j_1, j_2, \dots, j_k \neq i} \frac{|A_i| - |A_i \cap \left(\cup_{\ell=1}^k A_{j_\ell}\right)|}{|A_i|} \\ 
        & = \min_{j_1, j_2, \dots, j_k \neq i} \frac{|A_i \setminus \left(\cup_{\ell=1}^k A_{j_\ell}\right)|}{|A_i|} \geq  \frac{1}{n}.
\end{align*} 

The first inequality observes that the gap only worsens when we allow for \emph{all other} points to be used, rather than just those that come before $i$. The second inequality observes that, since all vectors inside the argument have integral coordinates, the output is the indicator vector over the union of the inputs. The third inequality observes that for any two sets $S, T,$ $\vec{e}_{S} \cdot \vec{e}_{T} = |S \cap T|$. The fourth inequality restates the previous line. The last inequality uses  $|A_i| \leq n$ in the denominator and the fact that $\mathcal{F}^L$ is $k$-cover free, thus $A_i \setminus \left(\cup_{\ell=1}^k A_{j_\ell}\right) \neq \emptyset$ for any choice of $A_{j_\ell}$ in the numerator. 

Thus, $\gaptwo{k}{i}{X^L}{Q^L} \geq 1/n$ for all $i \in \mathcal{F}^L$. Therefore, $\MenuGapTwo{k}{X^L}{Q^L} \geq \frac{|\mathcal{F}^L|}{n}$.
\end{proof}

\subsubsection{Part 3: from Sequences to Distributions}

We now present the final piece for the proof of Theorem~\ref{thm:mainLB}. Lemma~\ref{lem:hnconverse} states that given a pair of sequences $(X,Q)$ of a certain form, we can find a distribution whose revenue gap is at least as large as $\MenuGapTwo{k}{X}{Q}$. This is a slight generalization of a lemma from cite~\cite{HartN17}. The experienced reader will notice that our construction uses many similar ideas. Their work makes no assumptions on the sequences $X,Q,$ but only works for the case of $k=1$.

\begin{lemma}
\label{lem:hnconverse}
Let $(X,Q)$ be a pair of sequences such that $\vec{x}_i \in \{0,1\}^n, \vec{q}_i \in \{0,1\}^n$ for all $i$. Moreover, suppose $\gaptwo{k}{i}{X}{Q} \geq \frac{1}{n}$ for all $i$. Then, there exists a distribution $\distD \in \mathbb{R}^n$ such that for any integer $k$,

\begin{equation*}
    \frac{\bkrev{\distD}{k}}{\brev{\distD}} \geq \frac{\MenuGapTwo{k}{X}{Q}}{2n}.
\end{equation*}
\end{lemma}

\begin{proof}

In order to construct a distribution $\distD$ we need both a valuation $\hat{x}$ and a density function $f(hat{x})$. Let $C_i = (n+1)^{2i}$. Then we define distribution $\distD$ by setting $\vec{\hat{x}}_i = \vec{x}_i \cdot C_i$, $f(\vec{\hat{x}}_i) = \frac{1}{C^i}$. It is clear that this defines a valid distribution, i.e., $f(\vec{\hat{x}}_i) \geq 0$ and $\sum_{i} f(\vec{\hat{x}}_i) \leq 1$. Place the rest of the probability mass at a valuation of $\vec{0}^n$.  

We will now show that $\bkrev{\distD}{k} \geq \MenuGapTwo{k}{X}{Q}$ and $\brev{\distD} \leq 2n$. Consider the menu $\mathcal{M}$ which offers allocation $\vec{q}_i$ at price $p_i =  \gaptwo{k}{i}{X^L}{Q^L} \cdot C_i$. Let $\mathcal{M}_i$ be the sub-menu of $\mathcal{M}$ consisting of the first $i$ menu entries and the $(0,\vec{0}^n)$ entry. We will first claim that $\mathcal{M}$ is a buy-$k$ menu. 

\begin{claim}
Any valuation $\hat{x}_i$ prefers to purchase the menu entry $(p_i, \vec{q}_i)$ to any other combination of $k$ menu entries from $\mathcal{M}_{i}$.
\end{claim}

\begin{proof}
First, observe that if the valuation $\hat{x}_i$ purchases at least one copy of $(p_i, \vec{q}_i)$, because $\vec{x}_i = \vec{q}_i$ and $\vec{q}_i \in \{0,1\}^n$, there is no value in purchasing any other menu entry. This is because a buyer with valuation $\hat{x}_i$ is only interested in the items in the support of $\vec{q}_i$, all of which are given to the buyer with probability $1$. There is no benefit from purchasing any other lottery. Thus, any other reasonable deviations involve buying up to $k$ menu entries from $\mathcal{M}_{i-1}$. The utility from purchasing any such combination is upper bounded by $\hat{x}_i \cdot \lottery(\vec{q}_{i_1}, \vec{q}_{i_2}, \dots, \vec{q}_{i_k}).$ The utility from purchasing $(p_i, \vec{q}_i)$ is $\vec{\hat{x}}_i \cdot \vec{q}_i - p_i$. By choice of $p_i$, $\vec{\hat{x}}_i$ we get that this is 
$$C_i \vec{x}_i \cdot \vec{q}_i - C_i \cdot \gaptwo{k}{i}{X^L}{Q^L} \geq C_i \vec{x}_i \cdot \lottery(\vec{q}_{i_1}, \vec{q}_{i_2}, \dots, \vec{q}_{i_k}),$$
where the inequality follows from recalling the definition of $\gaptwo{i}{k}{X}{Q}$ (and cancelling the $C_i$). 
\end{proof}

The next thing we need to show is that the valuation will not prefer to buy any option on $\mathcal{M}\setminus \mathcal{M}_i$. The utility from purchasing the preferred option is at most $C_i \cdot n$. The cost of any further option is at least $C_{i+1} \cdot \gaptwo{k}{i+1}{X^L}{Q^L}$. By assumption, $\gaptwo{k}{i}{X}{Q} \geq \frac{1}{n}$. Therefore, the price of any option with $j > i$ is at least $C_{i+1}/n$. By construction, the price alone for any option $(p_j,\vec{q}_j)$ with $j > i$ is already greater than the possible utility the buyer could get. Thus, purchasing such menu entries would give them non-positive utility. Therefore, a valuation $\hat{x}_i$ will purchase exactly one copy of the menu entry $(p_i,\vec{q}_i)$. The revenue of mechanism $\mathcal{M}$ is $\sum_{i} f(\hat{x}_i) p_i = \sum_{i} \gaptwo{k}{i}{X}{Q} = \MenuGapTwo{k}{X}{Q}$. Since $\mathcal{M}$ is a buy-$k$ menu, $\bkrev{\distD}{k} \geq \bkrev{\distD, \mathcal{M}}{k}$.

All that remains is to show that the revenue of bundling is at most $2n$. Note that the value a valuation $\hat{x}_i$ has for the bundle is at most $n C_i \leq C_{i+1}$. Thus, any price between $(C_{i-1} \cdot |\vec{x}_{i-1}|, C_i \cdot |\vec{x}_i|]$ will sell to the same set of bidders. Since we want to maximize revenue, it only makes sense to consider prices $b_i = C_i \cdot |\vec{x}_i|$ for all $i$. Consider any such price $b_i$ for the bundle. The revenue is $b_i \cdot \Pr_{\hat{x}_j\sim \distD}( C_j |\vec{x}_j| \geq b_i) = b_i \cdot \Pr_{\hat{x}_j\sim \distD}(C_j \geq C_i) = b_i \cdot \sum_{j \geq i} f(\hat{x}_j) = b_i \cdot \frac{2n}{C_i(2n-1)} \leq 2 b_i \cdot \frac{1}{C_i} \leq 2n.$
\end{proof}

\subsubsection{Part 4: Putting it all together} 

We are now ready to present the Proof of Theorem~\ref{thm:mainLB}.

\begin{proof}[Proof of Theorem~\ref{thm:mainLB}]

Consider the instance $(X^L,Q^L)$ from Definition~\ref{def:instance}. By Claim~\ref{cl:expmenugap}, the instance satisfies $\MenuGapTwo{k}{X^L}{Q^L} \geq |\mathcal{F}^L|/n$ and has only integral vectors. By Lemma~\ref{lem:hnconverse}, we can turn the pair of sequences into a distribution $\distD$ with $\bkrev{\distD}{k}/\brev{\distD} \geq \frac{\MenuGapTwo{k}{X^L}{Q^L}}{2n}$. Thus, 

\begin{equation*}
    \frac{\bkrev{\distD}{k}}{\brev{\distD}} \geq \frac{|\mathcal{F}^L|}{2n^2}. 
\end{equation*}
Finally, for $k \leq n^{1/2-\varepsilon},$ note that $|\mathcal{F}^L| = \exp(\Omega({n^\varepsilon}))$. Therefore, the right hand side becomes $\frac{\exp(\Omega({n^\varepsilon}))}{2n^2}$.
\end{proof}

\section{Conclusion}
\label{sec:conclusion}

In this paper we initiate the study of fine-grained buy-many mechanisms. The motivation for our work stems from a simple observation: there exist distributions for which the buy-one revenue gap $\rev{\distD}/\brev{\distD}$ is unbounded, but for all distributions the buy-many revenue gap $\bmrev{\distD}/\brev{\distD}$ is finite. There is a wide worst-case revenue gap between optimal buy-many and optimal buy-one mechanisms, which begs the question: how much must we constraint the seller's choice of mechanism until the revenue gap becomes finite for all distributions? In order to answer this question, we introduce the concept of buy-$k$ mechanisms, those where the buyer can buy any multi-set of up to $k$ many menu choices. We show that buy-$n$ mechanisms are not much better than bundling. For all distributions $\distD$, the revenue from bundling recovers a $O(n^2)$ fraction of the optimal buy-$n$ revenue when the buyer is additive and a $O(2^n\cdot n^2)$ fraction of the optimal buy-$n$ revenue when the buyer has an arbitrary monotone valuation. Our proof uses a recent framework proposed in~\cite{HartN17, PSW22} for buy-one mechanisms. While in those works, the framework has been used to produce examples of inapproximable distributions, our work shows that it can be used to prove approximation guarantees. Moreover, all our results hold for the case of an adaptive buyer. 

There are numerous questions for future work:
\begin{itemize}
    \item We have already outlined one interesting question for future work, to prove or disprove Conjecture~\ref{app:conj}. Even a slight weakening of the conjecture, just showing that for every $k$ there exists some $\distD_k$ such that $\bkrev{\distD_k}{k} > \bkrev{\distD_k}{k+1}$ would be interesting. We discuss our candidate instance for Conjecture~\ref{conj} in Appendix~\ref{app:conj}. 
    \item We showed that restricting the seller to be buy-$n$ incentive compatible sufficed to obtain a $O(2^n\cdot n^2)$-approximation via bundling. It would be interesting if the exponential bound is tight in general; if it is, it would be interesting to characterize the distributions for which a polynomial approximation is possible.
    \item It would be interesting to understand the role of $k$ in whether or not the revenue gap is finite. Concretely, we would like to answer the following question: for a given $n$, what is the smallest $k$ for which $\bkrev{\distD}{k}/\brev{\distD}$ is finite for all $\distD$? Ultimately, we would like to understand the exact trade off between $k, n$ in the revenue gap, answering Open Question~\ref{oq:1}.
    \item Future work could also follow the steps of~\cite{CTT20} in understanding whether or not fine-grained buy-many mechanisms satisfy revenue monotonicity, or whether or not fine-grained buy-many mechanisms admit $(1-\varepsilon)$-approximations via finite-sized mechanisms (and what role $k$ has in answering any of these questions).  
    \item Another interesting avenue would be to explore the power that buy-many or fine-grained buy-many mechanisms have over product distributions. There is a long line of work with elegant approximation results for the case of product distributions, but progress towards polynomial time approximation schemes has been slow. It is possible that restricting the seller's choice of mechanism improves the performance of existing algorithms or allows for the discovery of more efficient ones.  
    \item Finally, very little is understood computationally about buy-many and fine-grained buy-many mechanisms. For instance, it is not immediately clear how to efficiently test whether or not a mechanism is buy-many or buy-$k$ for some $k$. 
\end{itemize}

We hope that our results strengthen the importance of developing a deeper understanding of fine-grained buy-many mechanisms.


\bibliographystyle{ACM-Reference-Format}
\bibliography{references.bib}

\appendix

\section{Adaptive Buy-Many Mechanisms} 
\label{sec:adapt}

In this Appendix, we briefly review another notion of buy-$k$ mechanisms, which we refer to as \emph{adaptive} buy-$k$ mechanisms. We will define them to be analogues of the adaptive buy-many mechanisms as defined in~\cite{CTT19}. 

In the standard definition of buy-$k$ mechanisms, formalized in section~\ref{sec:Notation}, the buyer may purchase any multi-set of menu options of size up to $k$. In a randomized mechanism, this corresponds to committing to up to $k$ options and \emph{only after that} receiving their outcome allocations; in other words, the \emph{choice} of, say, second option, is \emph{not} a function of probabilistic outcomes of the lottery for the first option. This is formally captured in our definition of the function $\lottery(\cdot)$.

We can naturally also consider a variant of this definition that allows for \emph{adaptively} choosing the options to purchase, based on the probabilistic outcomes of the lotteries for the prior options. In this case, the buyer can commit to a \emph{strategy} of different ways of purchasing up to $k$ options, while seeing the outcome of each purchased lottery before purchasing the next option. A strategy can be thought of as a $2^n$-ary tree of depth at most $k$ where each node identifies what to purchase on the next step depending on which items of the current purchased lottery ``succeeded'' or ``failed''. The buyer is then interested in a strategy with maximum expected payoff. Analogous to~\cite{CTT19}, we say a mechanism $\mech$ is \emph{adaptive buy-$k$ incentive-compatible} if for every valuation of the buyer, the strategy with maximum expected payoff consists of buying a single option (see also Section 2 of ~\cite{CTT19} for more details on this definition).

As was observed in~\cite{CTT19}, it is easy to see that since the set of non-adaptive buy-$k$ options are all valid strategies for an adaptive buy-$k$ mechanism, any mechanism that is adaptive buy-$k$ incentive-compatible is also (non-adaptive) buy-$k$ incentive-compatible (but the reverse direction is not necessarily true). As a corollary of this, we can immediately extend our bounds in Theorems~\ref{thm:MainN} and~\ref{thm:GeneralVal} to adaptive buy-$k$ incentive-compatible mechanisms. 

Finally recall that the construction of Theorem~\ref{thm:mainLB} presented in Section~\ref{sec:few} gave a deterministic mechanism. When a mechanism is deterministic, there is no distinction between adaptive and non-adaptive strategies because there is no randomness in the allocation. Therefore, the lower bounds of Theorem~\ref{thm:mainLB} also extend to the case of adaptive buyers.

\section{$\brev{\distD}$ Can Not Give a Sublinear Approximation to $\bkrev{\distD}{k}$}
\label{app:polyn}

In this section we show that Theorem~\ref{thm:MainN} can not be improved to a sublinear approximation factor for any $k$. 

\begin{claim}
\label{cl:polyappx}
There exists a distribution $\distD$ such that for a single, additive buyer and any $k$, 

$$ \brev{\distD} \leq \frac{2\cdot \bkrev{\distD}{k}}{n}.$$
\end{claim}

\begin{proof} 
Consider the distribution $\distD$ from Example 15 of~\cite{HartN12}. This distribution satisfies the following property: $\brev{\distD} = 2, \srev{\distD} = n$, where $\srev{\cdot}$ is the optimal revenue attained by item-pricing. Observe that item-pricing is a buy-many incentive-compatible mechanism. Thus, $\srev{\distD} \leq \bmrev{\distD}$. Moreover, from Claim~\ref{cl:bmrev} we know that for any $k$, $\bmrev{\distD} \leq \bkrev{\distD}{k}$. Therefore, $$ \brev{\distD} \leq \frac{2\cdot \bkrev{\distD}{k}}{n}.$$
In particular, this shows that the ratio between $\brev{\distD}$ and $\bkrev{\distD}{k}$ for additive buyers is $\Omega(n)$, implying that Theorem~\ref{thm:MainN} can not be substantially improved.
\end{proof}

\section{Candidate Distribution for a Separation Between $\bkrev{\distD}{k}$ and $\bkrev{\distD}{k+1}$} 
\label{app:conj}

In this section we present the instance $\distD$ that proves Proposition~\ref{prop:separation} and that we posit could prove Conjecture~\ref{conj}. 

\begin{example}
\label{ex:conj}
Consider the following (correlated) distribution $\distD$ over two additive items:

\[   
\Pr(v_1 = a, v_2 = b) = 
     \begin{cases}
       1/6 &\quad\text{for } a = 3, b = 4\\
       1/6 &\quad\text{for } a = 4, b = 3\\
       4/6 &\quad\text{for } a = 5, b = 7.\\
     \end{cases}
\]
\end{example}
For the following proofs, we introduce some additional notation. Namely, we denote each $x_j \in \textsc{supp}({\distD})$ as ``buyer $j$” and let $u_{j}(\Lambda) = {x_j}(\lottery(\Lambda)) - \sum_{i=1} p(\vec{q}_i)$ be the utility gained by buyer $j$ from purchasing the multi-set of allocations $\Lambda$. We also prove the following useful claim about the utility function. 
\begin{claim}
\label{cl:lotconcave}
Let buyer $j$ have a valuation of $x_j$. Let $\Lambda$ be some multi-set of allocations such that $q_i \in \Lambda$ with multiplicity $m_i \geq 0$. Then, $u_{j}(\Lambda)$ is concave in $m_i$.
\end{claim}
\begin{proof}
We verify that the second derivative of the utility function for buyer $j$ with respect to $m_i$ is non-positive.
\begin{align*}
    \frac{\partial^2 u_j(\Lambda)}{\partial m_i^2} &= \frac{\partial^2}{\partial m_i^2}\left({x_j}(\lottery(\Lambda)) - \sum_{i=1} p(\vec{q}_i)\right)\\
    &= \frac{\partial^2 }{\partial m_i^2}\left(\sum_{t = 1} x_{jt} \cdot (1-(1-q_{1t})^{m_1}\cdot ... \cdot(1-q_{it})^{m_i}) - \sum_{i=1} p(\vec{q}_i)\right)\\
    &= \sum_{t = 1} -x_{jt} \cdot \ln{(1-q_{it})}^2\cdot(1-q_{1t})^{m_1}\cdot ... \cdot(1-q_{it})^{m_i}\\
    &\leq 0
\end{align*}
since $\Vec{q}_i \in [0,1]$ and $x_j \geq 0$.
\end{proof}

\begin{proof}[Proof of Proposition~\ref{prop:separation}]
The LP formulation presented in \cite{BriestCKW15} provides us with a simple way to compute optimal mechanisms in the buy-one world. Naturally, for some fixed input distribution, the LP constructs an allocation and payment function which maximizes expected revenue while maintaining feasibility and buy-one incentive-compatibility constraints. Through this LP, we can compute the revenue-optimal buy-one mechanism for the distribution $\distD$ shown below:
\begin{align*}
    \mech_1 = 
    \begin{cases}
        ((0.2,0.2), 1.4)\\
        ((1,0), 4)\\
        ((1,1), 11).\\
    \end{cases}
    \quad
\end{align*}
This mechanism achieves an expected revenue (subject to truncation) of $8.233$. More generally, the previous program can be altered to compute the optimal buy-$k$ mechanism for a distribution $\distD$, albeit by introducing non-convex constraints to enforce buy-$k$ incentive-compatibility. Despite non-convexity, for $k=2,3,4$, a non-convex optimizer was able to compute the revenue-optimal buy-$k$ mechanism. The annotated code samples for computing these optimal mechanisms can be found \href{https://github.com/vikram-kher/Computing-Revenue-Optimal-Buy-k-Auctions}{here}. For each respective value of $k$, the optimizer found the optimal expected revenue to be $[8.135, 8.096, 8.074]$, allowing us to conclude that $\bkrev{\distD}{1} > \bkrev{\distD}{2} > \bkrev{\distD}{3} > \bkrev{\distD}{4}$. Thus, we empirically validate Proposition~\ref{prop:separation}.
\end{proof}

We also conjecture that Example~\ref{ex:conj} is a good candidate for proving Conjecture~\ref{conj}. Fix some value of $k \geq 1$ and consider the following mechanism $\mech_k = \{t_1 =( (\alpha_k, \alpha_k), 7 \alpha_k), t_2 = ((1, 0), 4), t_3 = ((1, 1), 11)\}$, where $\alpha_k$ is the smallest positive real root of the $k$-degree polynomial $f_k(x) = 12 (1-(1-x)^k) - 7 k \cdot x - 1$. Intuitively, the value $\alpha_k$ has the following property: when buyer 3 purchases $k$ copies of $t_1$, they receive a utility of $1$, which is the same utility gained from purchasing a single copy of $t_3$.

\begin{claim}
\label{cl:conj1}
The mechanism $\mech_k$ is buy-$k$ incentive compatible for the distribution defined in Example~\ref{ex:conj}, and achieves expected revenue $R_k = 8 + 7/6 \cdot \alpha_k$.
\end{claim}

\begin{proof}
We first verify that $\mech_k$ is a buy-$k$ incentive-compatible mechanism for the distribution $\distD$. By Claim \ref{cl:lotconcave} the utility of buyer $j$ from purchasing a multi-set of allocations $\Lambda$ is concave in the number of identical allocations bought. Consequently, once we establish that a buyer achieves non-positive utility from purchasing a multi-set $\Lambda$ of allocations, we can conclude that purchasing any multi-set $\Lambda' \supseteq \Lambda$ will similarly yield non-positive utility for the buyer.

We proceed by calculating the utility received by each buyer for each allocation and showing that each buyer (weakly) maximizes their utility by purchasing a single ticket. First, consider buyer 1 who achieves $u_1((\alpha_k,\alpha_k)) = 0$, $u_1((1,0)) < 0$, and $u_1((1,1)) < 0$. Second, consider buyer 2 who achieves $u_2((\alpha_k,\alpha_k)) = 0$, $u_2((1,0)) = 0$, and $u_2((1,1)) < 0$. By Claim \ref{cl:lotconcave}, both buyer 1 and buyer 2 satisfy the buy-$k$ incentive-compatibility constraints since they will always prefer to purchase tickets $t_1$ and $t_2$, respectively, compared to another other multi-set of allocations.

For buyer 3, a slightly more careful analysis is required. We can first check that buyer 3 satisfies buy-one incentive-compatibility constraints since $0 <  u_3((\alpha_k,\alpha_k)) \leq 1$, $u_3((1,0)) = 1$, and $u_3((1,1)) = 1$. We must additionally check that the higher order incentive-compatibility constraints also hold since buyer $3$ obtains positive utility from individually buying $t_1$ and $t_2$. Notice, $t_2$ offers a deterministic allocation, so purchasing multiple copies of this ticket does not yield additional utility. This leaves two cases to analyze. Specifically, we can easily verify that $u_3(\{(\alpha_k,\alpha_k),(1,0)\}) = 1$ and $u_3(\{(\alpha_k,\alpha_k)_1,...,(\alpha_k,\alpha_k)_k\})= 1$, where the last equality follows directly from the definition of $\alpha_k$. Moreover, since $\alpha_k$ is the smallest positive root of the polynomial $f_k(x)$, we know that $u_3(\{(\alpha_k,\alpha_k)_1,...,(\alpha_k,\alpha_k)_i\}) < 1$ for $1 \leq i < k$. By Claim \ref{cl:lotconcave}, we can conclude that buyer 3 can achieve a utility of at most 1 by deviating from the ticket $t_3$. Thus, buyer 3 satisfies buy-$k$ incentive-compatibility constraints.

The revenue of $M_k$ follows from the fact that buyer $j$ will purchase ticket $t_j$ and from the density of the valuation classes. In the case where two allocations yield the same utility for a buyer, we break the tie in favor of the seller. As a result, $R_k = 2/3 \cdot 11 + 1/6 \cdot 4 + 1/6 \cdot 7\alpha = 8 + 7/6 \cdot \alpha_k$.
\end{proof}

\begin{claim}
\label{cl:conj2}
For all integers $k \geq 1$, the polynomial $f_k(x)$ always has a positive, real root below $1$. Moreover, the sequence  $\{\alpha_k\}_{k=1}^{\infty}$, where $\alpha_k$ is the smallest positive, real root of $f_k(x)$, is strictly decreasing.
\end{claim}

\begin{proof}
The proof follows via an inductive argument using the intermediate value theorem. First, recognize that $f$ is continuous and $f_k(0) = -1$ for all $k \geq 1$. By the proof of Proposition~\ref{prop:separation}, we have that $f_1(0.2) = 0 < 1$. Let us inductively assume that $f_k(\alpha_k) = 0$ and $\alpha_k \leq 1$ for $k \geq 1$. We wish to show that $f_{k+1}(\alpha_k) > 0$ as this would imply that $f_{k+1}(x)$ has a root $0 < \alpha_{k+1} < \alpha_k$ by the intermediate value theorem. To begin,
\begin{align*}
    f_{k+1}(\alpha_k) &= -12(1-(1-\alpha_k)^{k+1})\ -7\alpha_k\cdot (k+1)\ - 1\\
    &= -12(1-\alpha_k)^{k+1} + 12(1-\alpha_k)^k - 7\alpha_k\\
    &= -12(1-\alpha_k)^{k}\alpha_k - 7 \alpha_k,
\end{align*}
where the second equality follows from the fact that $\alpha_k$ is a root of $f_{k}(x)$. From the reduced form above, it suffices to show that $\alpha_k < (1 - (\frac{7}{12})^{1/k})$ to prove that $f_{k+1}(\alpha_k) > 0$. This can be accomplished by the following algebraic manipulations:
\begin{align*}
    f_k(1 - (7/12)^{1/k}) &= -12((7/12)^{1/k})^{k} -7k\cdot(1-(7/12)^{1/k}) + 11\\
    &= -7k \cdot(1-(7/12)^{1/k})+4\\
    &> 0.
\end{align*}
Since $f_k(1 - (7/12)^{1/k}) > 0$, by the intermediate value theorem, it follows that $\alpha_k < (1 - (\frac{7}{12})^{1/k})$. Consequently, we find that $f_{k+1}(\alpha_k) > 0$, finishing the proof that $0 < \alpha_{k+1} < \alpha_k < 1$. 
\end{proof}
Note that Claim~\ref{cl:conj2} clearly implies that the sequence $\{R_k\}_{k=1}^{\infty}$ is also strictly decreasing as it is the same sequence as $\{c_k\}_{k=1}^{\infty}$ but shifted by a constant. Given Claims~\ref{cl:conj1},~\ref{cl:conj2}, all that remains is to show that $\bkrev{\distD}{k} = R_k$ for all $k \geq 2$. Then since $\bkrev{\distD}{1} > R_2$ and the sequence $R_k$ is strictly decreasing, we would have that $\bkrev{\distD}{k} > \bkrev{\distD}{k+1}$ for all $k$, proving Conjecture~\ref{conj}. The principal obstacle to proving the missing step is find a technique to overcome non-convexity of the optimization problem. We currently have a candidate mechanism, but the non-convexity of the constraints hinders our ability to prove it is optimal via some notion of duality.

\end{document}